\documentclass[a4paper,UKenglish]{lipics-v2016-arxiv}
 
\usepackage{microtype}

\title{Courcelle's Theorem Made Dynamic\footnote{This 
    work is supported by EU under ERC
    EQualIS (FP7-308087) and STREP Cassting (FP7-601148).}}

\author[1]{Patricia Bouyer}
\author[1]{Vincent Jug\'e}
\author[1,2]{Nicolas Markey}
\affil[1]{LSV, CNRS \& ENS Cachan, Univ. Paris-Saclay, France}
\affil[2]{IRISA, CNRS \& Inria \& Univ. Rennes 1, France}
\authorrunning{P. Bouyer, V. Jug\'e and N. Markey} 

\Copyright{Patricia Bouyer, Vincent Jug\'e and Nicolas Markey}

\EventEditors{}
\EventNoEds{0}
\EventLongTitle{}
\EventShortTitle{}
\EventAcronym{}
\EventYear{}
\EventDate{}
\EventLocation{}
\EventLogo{}
\SeriesVolume{}
\ArticleNo{}

\usepackage{tikz}
\usetikzlibrary{arrows,backgrounds,calc}
\usepackage{tkz-graph}
\usepackage{enumerate}
\usepackage{icalp-arxiv}
\usepackage{thm-restate}
\usepackage{amssymb}

\begin{document}

\maketitle

\begin{abstract}
Dynamic complexity is concerned with updating the output of a problem
when the input is slightly changed. We~study the dynamic complexity of
model checking a fixed monadic second-order formula over
evolving subgraphs of a fixed maximal graph having bounded tree-width;
here the subgraph evolves by losing or gaining edges (from the maximal graph).
We show that this problem is in \Dyn\FO (with
\LOGSPACE precomputation), via a reduction to a
Dyck reachability problem on an acyclic automaton.

\end{abstract}

\section{Introduction}
\label{sec-intro}
\subparagraph{Monadic second-order logic, tree-width of graphs, and Courcelle's theorem.}
Monadic second-order logic (\MSO) is a powerful formalism for
expressing and checking properties of graphs. It~allows first-order
quantification (over states of the graph) as well as \emph{monadic}
second-order quantification (over sets of states), and can thus
express properties such as connectivity or 3-colorability of finite
graphs. While the satisfiability problem of this logic is in general
undecidable, model checking (i.e., deciding if a formula holds true in
a given graph) is \PSPACE-complete~\cite{Var82}; when the \MSO formula is fixed,
the problem is hard for every level of the polynomial hierarchy over finite graphs~\cite{Sto76},
and can be performed in \PTIME over finite trees.

The tree-width of graphs has been defined by Robertson and
Seymour~\cite{RS84} as a measure of the \emph{complexity} of
graphs---intuitively, of how close a graph is to a tree. Many classes
of graphs have been shown to have \emph{bounded} tree-width~\cite{Bod93}.
Over such graphs, several (\NP-)hard problems can be solved
in polynomial time~\cite{Bod93}.  Model checking a fixed \MSO
formula is such an example, as proven by Courcelle in~1990~\cite{Cou90}.

\subparagraph{Dynamic problems and dynamic complexity.} 
In this paper, we focus on the \emph{dynamic complexity} of
\MSO model checking over finite graphs. Dynamic-complexity theory aims at developing algorithms
that are capable of efficiently updating the output of a problem after
a slight change in its input~\cite{DST95,PI97}.  Such algorithms would
keep track of auxiliary information about the current instance, and
update it efficiently when the instance is modified.
Consider the problem of reachability in directed graphs,
and equip such graphs with two operations, for respectively inserting
and deleting edges (one at a time). It has recently been proven that this problem
is in the class \Dyn\FO~\cite{DKMSZ15}, which was a
long-standing open problem.  Roughly speaking, a problem is in
\Dyn\FO when
it admits an algorithm updating
the solution and some auxiliary information
through \FO formulas (or, equivalently by $\AC^0$ circuits
satisfying some uniformity requirements)
after a small change in the input.

\subparagraph{Our contributions.}
We study the \MSO model-checking problem from a dynamic perspective,
considering the 
following basic operations on graphs: insertion and deletion of an edge.
We~assume that we are given a maximal graph, which embeds all
constructed game graphs along the dynamic process: this maximal graph
represents the set of all possible connections in the subgraphs we
will consider.
We first realize that, since the \MSO model-checking problem
over arbitrary finite graphs is \NP-hard (see~\cite{Kar72,Sto76}),
the \MSO model-checking problem over arbitrary
maximal graphs is unlikely to be solvable in \Dyn\FO,
even allowing \PTIME precomputation (unless \Dyn(\PTIME,\FO)=\PTIME=\NP).

We therefore make a standard restriction and assume that the maximal
graph has bounded tree-width. Under this hypothesis, we show that
the \MSO model checking over such graphs
can be solved in \Dyn\FO with \LOGSPACE precomputation.
To obtain this result, we rely on (and extend) a \Dyn\FO
algorithm for finding a Dyck path in an acyclic automaton~\cite{WS07},
and build a transformation of our model checking problem into such a
Dyck reachability problem. The latter transformation is performed by
using Courcelle's theorem, and by realizing that the runs of
bottom-up, deterministic tree automata can be computed step-by-step.
These simple steps can be stored in an
auxiliary graph, in which only few edges depend on the real edges
that exist in the original game; the correctness of the
construction then goes through the search for paths labeled with Dyck
words.

\subparagraph{Related works.}
\MSO has been extensively studied over various classes of structures
in the last 40 years, both regarding its expressiveness and regarding
the algorithmic properties of its decision problems (see~\cite{CE11}
and references therein). Similarly, numerous measures of the
\emph{complexity} of graphs, such as tree-width~\cite{RS84},
clique-width~\cite{CER93b} or entanglement~\cite{BG05}, have been
defined and studied; they provide large classes of graphs in which
different kinds of hard problems become tractable (see~\cite{Rab08}
and references therein).

On the other hand, dynamic complexity is much less developed: while
the main dynamic complexity classes were defined and studied 20 years
ago~\cite{DST95,PI97},
only few problems have been
considered from that point of view~\cite{WS07,DKMSZ15,Zeu15}.
As cited above, directed-graph
reachability has recently been proven in \Dyn\FO, which was an important
open problem in the area.

Finally, let us mention that the results reported in this paper were originally presented in the
setting of parity games played on a graph having
bounded tree-width~\cite{arxiv}.%

\section{Definitions}
\label{sec-defs}
\subsection{Monadic second-order logic over graphs}

A graph is a pair $G = \tuple{V,E}$ where $V$ is a finite set of
vertices, and $E \subseteq V \times V$ is a finite set of edges. The
size of~$G$ is the cardinality of~$V$. Formulas of the monadic
second-order logic (denoted \MSO) are built using first-order and
(monadic) second-order variables, used respectively to quantify over
vertices and sets of vertices; formulas may also use equality, and the
edge relation~$E$ of the graph. As~an example, \emph{connectivity} of
a graph can be expressed as
\[
\forall S.\ \left[
  (\forall x.\ x\in S) \ou
  (\forall x.\ x\notin S) \ou
  (\exists x,y.\ (x\in S \et y\notin S \et E(x,y))
  \right].
\]

The standard static question regarding \MSO over graphs is to decide
whether a given formula is satisfiable, or to check whether it is
satisfied in a given graph model. These problems have been extensively
considered in the literature; in particular, satisfiability is
undecidable~\cite{See76}, while model checking is \PSPACE-complete~\cite{Var82}.
We~refer to~\cite{Kre08} for a survey.

\subsection{Tree decomposition}

The notion of \emph{tree decomposition}~\cite{RS84,RS86} was
introduced by Robertson and 
Seymour. It~gives rise to classes
of graphs on which many problems that are \NP-hard in general become
tractable.

\begin{dfn}
An ordered tree decomposition of $G$ is
a pair $\calD = \tuple{\calT,\bfT}$, where $\calT =
\tuple{\calN,\calE}$ is an ordered tree, and $\bfT \colon
\calN \to 2^V$ is a function such that:
\begin{itemize}
\item for each edge $e \in E$, there exists a node $n \in
  \calN$ such that $e \in \bfT(n)^2$;
\item for each vertex $v \in V$, the set
  $\calN_v = \{n \in \calN \mid v \in \bfT(n)\}$ is
  non-empty, and the restriction of~$\calT$ to~$\calN_v$ is connected.
\end{itemize}
The \emph{width} of $\calD$ is defined as the integer
$\max\{|\bfT(n)|\mid n \in \calN\}-1$, and the
\emph{tree-width} of $G$ is the least width of all tree decompositions of~$G$.
\end{dfn}

\subsection{Tree automaton}

The notion of (deterministic, bottom-up) \emph{tree automaton}
is a powerful tool for
expressing and checking properties of finite trees.

\begin{dfn}
A \emph{tree automaton} is a tuple $\calA =
\tuple{Q,\Sigma,\iota,Q_{\END},\delta}$ where $Q$~is a finite set of
states, $\Sigma$~is a finite input alphabet, $\iota \in
Q$ is the initial state, $Q_{\END} \subseteq Q$ is the set of
accepting states and $\delta \colon Q^2 \times \Sigma \to Q$ is the
transition function.

Let $\calT = \tuple{\calN,\calE}$ be a binary ordered labeled tree, with label set~$\Sigma$.
The~\emph{run} of~$\calA$ over~$\calT$ is the function $\rho \colon \calN \to Q$ such that:
\begin{itemize}
 \item for every leaf~$n$ of~$\calT$ with label~$\lambda$, we~have
   $\rho(n) = \delta(\iota,\iota,\lambda)$;
 \item for every internal node~$n$ of~$\calT$ with label~$\lambda$ and
   with children $m_1$ and~$m_2$, we have $\rho(n) =
   \delta_2(\rho(m_1),\rho(m_2),\lambda)$.
\end{itemize}
If, furthermore, the~run~$\rho$ maps the root of~$\calT$ to an accepting state
$q \in Q_{\END}$, then we say that $\rho$~is \emph{accepting},
and that the automaton~$\calA$ \emph{accepts} the tree~$\calT$.
\end{dfn}

\subsection{Dynamic complexity theory}

In this paper, we adapt Courcelle's theorem to a dynamic-complexity
framework. We~briefly introduce the formalisms of
descriptive- and dynamic complexity here, and refer to~\cite{PI97,Imm99,Hes03,WS07}
for more details.

Descriptive complexity aims at characterizing positive instances of a
problem using logical formulas: the~input is then given as a
logical structure described by a set of $k$-ary predicates
(the~\emph{vocabulary}) over its universe. For example, a~directed
graph can be represented as a binary predicate representing its edges,
with the set of vertices (usually identified with $\{1,\ldots,n\}$
for some~$n$) as the universe. Whether
each vertex has at most one outgoing edge is expressed by the
first-order formula
$\forall x, y, z. (E(x,y) \et E(x,z)) \Rightarrow (y=z)$.
The~complexity class~\FO contains all problems that can be
characterized by such first-order formulas.  This class corresponds to
the circuit-complexity class~$\AC[0]$ (under adequate uniformity
assumptions)~\cite{BIS90}.

Dynamic complexity aims at developing algorithms that can efficiently
update the output of a problem (e.g. reachability of a given vertex in
a graph) when the input is slightly changed.
In~this setting, algorithms may take advantage of previous computations
in order to very quickly recompute the solution for the modified input.

Formally, following~\cite{WS07}, a decision problem~$\sfS$ is a
subset of the set of $\tau$-structures $\Struct{\tau}$ built on a
vocabulary~$\tau$. In~order to turn~$\sfS$ into a dynamic
problem~$\Dyn\sfS$, we~need to define a finite set of initial inputs
and a finite set of allowed updates. For instance, we might use an
arbitrary graph as initial input, then use a $2$-ary operator
$\textsf{ins}(x,y)$ that would insert an edge between vertices~$x$
and~$y$.

Hence, we associate the decision problem $\sfS$
with a set $\Updates$ of update functions
$\up\colon \Struct{\tau} \to \Struct{\tau}$.
We~identify every non-empty word in
${\Struct{\tau} \cdot \Updates^\ast}$ with the $\tau$-structure
obtained by applying a sequence of update operations
to an initial structure.
Denoting by $\STRUCT{n}{\tau}$ and by $\Updates_n$
the set of $\tau$-structures and of updates restricted to a
universe of size $n$, we define the dynamic language 
$\Dyn\sfS_n$ as the set of those words in $\STRUCT{n}{\tau} \cdot \Updates_n^\ast$
that correspond to structures of $\sfS$.
The dynamic language $\Dyn\sfS$ is then defined as the union
(over all~$n$) of all such languages.

Given two complexity classes $\calC$ and~$\calC'$,
a~dynamic problem $\Dyn\sfS$
with set of updates $\Updates$ belongs to the class
$\Dyn(\calC,\calC')$ if, and only~if,
there exists an auxiliary vocabulary~$\tau^{\aux}$,
a~$\calC$-computable \emph{initialisation} function
$f^{\INIT}\colon \Struct{\tau} \to \Struct{\tau^{\aux}}$,
a~$\calC'$-computable \emph{update} function
$f^{\up} \colon \Struct{\tau^{\aux}} \times \Updates \to \Struct{\tau^{\aux}}$,
and a $\calC'$-computable \emph{decision} function
$f^{\dec} \colon \Struct{\tau^{\aux}} \to \{0,1\}$ such that:
\begin{itemize}
 \item for every structure $A \in \Struct{\tau}$ and every
 upate $\up \in \Updates$, we have
 $f^{\INIT}(\up(A)) = f^{\up}(f^{\INIT}(A),\up)$;
 \item for every structure $A \in \Struct{\tau}$, we have
 $A \in \sfS \Leftrightarrow f^{\dec}(f^{\INIT}(A)) = 1$.
\end{itemize}
If, furthermore, $f^{\INIT}$ maps the empty structure of $\Struct{\tau}$
to the empty structure of $\Struct{\tau^{\aux}}$, then
we say that $\Dyn\sfS$ belongs to the class $\Dyn\calC'$.

Informally, $\Dyn\sfS$ belongs to $\Dyn(\calC,\calC')$
if, by~maintaining an auxiliary structure (which may
have an initial cost in~$\calC$),
an~algorithm can
tackle every update on the input structure
with a cost in~$\calC'$.
If~the initial cost is reduced to zero when the initial input
is the empty structure, then $\Dyn\sfS$ belongs to $\Dyn\calC'$.

In~this paper, we consider the case where $\calC=\LOGSPACE$ and $\calC' = \FO$,
meaning that precomputations will be carried out in~$\LOGSPACE$ and that
first-order formulas will be used to describe how
predicates are updated along transitions.

\subsection{Main result}

We are now in a position to formally define our problem and state our
main result. We~fix an \MSO formula~$\varphi$.
We follow the approach of~\cite{DG08a}, and
represent graphs as tuples $\tuple{V,E}$.
Given a universe $V$, our initial structure consists of a tuple
$\tuple{V,E_\star,E}$, where $E_\star$ is
a maximal set of edges and $E \subseteq E_\star$
is an initial set of edges.

We focus below on the operations of
insertion and deletion of edges that belong to $E_\star$. More
precisely, we~let $\Updates_{E_\star} = \{\textsf{ins}(e),
\textsf{del}(e) \mid e \in E_\star\}$. The effect of a
sequence of update operations, represented as a
word~$w\in\Updates_{E_\star}^\ast$, over a set~$E \subseteq E_\star$ of
edges, is denoted by $w(E)$, and is defined inductively~as:
  \[
  \begin{array}{ll>{\qquad\qquad\qquad}ll}
    E &\quad\text{ if $w=\epsilon$} & 
    E\cup\{e\} & \quad\text{ if $w=\textsf{ins}(e)$} \\
    w'(a(E)) &\quad\text{ if $w=w'\cdot a$} &
    E\setminus\{e\} &\quad \text{ if $w=\textsf{del}(e)$} 
  \end{array}
  \]
  For~$w\in \Updates_{E_\star}^{\ast}$ and $E \subseteq E_\star$,
  we~write $G_{w(E)}$ for the graph with vertex set $V$ and
  edge set~$w(E)$.
  It is to be noted that $G_{w(E)}$ is a subgraph
  of $\tuple{V,E_\star}$. Finally, we~let
  \[
  \Dyn\GraphMSO_\varphi =
  \{\tuple{V,E_\star,E} \cdot w \mid w\in\Updates_{E_\star}^\ast \text{ and } G_{w(E)} \vDash \varphi\}.
  \] 

  As mentioned in the introduction, the above problem is unlikely to
  be solvable in
  \Dyn(\PTIME,\FO).
  We therefore  adopt the idea of bounding the tree-width of the maximal graph $\tuple{V,E_\star}$.
  We fix a positive integer~$\kappa$ and restrict the set of admissible initial inputs:
  the graph $\tuple{V,E_\star}$ should be of tree-width at most $\kappa$.
  We~thus refine our problem as follows:
  \[
  \Dyn\GraphMSO_{\kappa,\varphi} = \{\tuple{V,E_\star,E} \cdot w \mid
  \tuple{V,E_\star} \text{ has tree-width at most } \kappa\} \cap
  \Dyn\GraphMSO_\varphi.
  \]
  Our main contribution is a dynamic algorithm for deciding
  $\Dyn\GraphMSO_{\kappa,\varphi}$.

  \begin{restatable}{thm}{mainthm}
    \label{mainresult}\label{thm-main}
    Fix a positive integer~$\kappa$ and an \MSO formula~$\varphi$.
    The problem $\Dyn\GraphMSO_{\kappa,\varphi}$
        can be solved in $\Dyn(\LOGSPACE,\FO)$.
\end{restatable}

We~give a short overview of the proof here. Our~algorithm consists in
transforming our \MSO model checking problem into an equivalent Dyck
reachability problem over a labeled acyclic graph. The~latter problem
is known to be in \Dyn\FO~\cite{WS07}, although~we~had to adapt the
algorithm to our setting.
Our~approach for building this acyclic graph follows from an
automata-based construction used for proving Courcelle's theorem:
along some linearization of a tree decomposition of the maximal graph,
we can inductively compute local information about the possible
computations of a bottom-up tree automaton.
These computations can be represented as finding a path in an acyclic
graph. However, we~have to resort to \emph{Dyck paths} in order to
make our acyclic graph efficiently updatable when the input graph
is modified.

\section{Courcelle's theorem}
\label{sec-treedec}
Courcelle's theorem is not based on working directly with tree
decompositions of graphs, but on labeled ordered trees whose labels
are chosen from a finite alphabet, and that represent such tree
decompositions. We begin with defining such trees.

\begin{dfn}
Let $G = \tuple{V,E}$ be a graph, and let $\calD = \tuple{\calT,\bfT}$
be a binary ordered tree decomposition of~$G$ of width~$\kappa$.
Let~$\Sigma = 2^{\{0,\ldots,\kappa\}} \times
2^{\{0,\ldots,\kappa\}} \times 2^{\{0,\ldots,\kappa\}^2}$ be a
(finite) set of labels.
We call \emph{proper $\calD$-coloring} of $G$ a function $\chi \colon
V \to \{0,\ldots,\kappa\}$ such that, for all nodes~$n$ of~$\calT$,
the~restriction of~$\chi$ to~$\bfT(n)$ is injective.  We~then call
\emph{$(\chi,\calD)$-succinct tree decomposition} of~$G$ (we~may omit to
mention~$\chi$ and $\calD$ if it is clear from the context)
the rooted tree obtained by labeling every node~$n$ of~$\calT$ with a
label $\lambda(n) = \tuple{\chi(A),\chi(B),\chi(C)} \in \Sigma$ as follows:
\begin{itemize}
  \item we set $A = \bfT(n) \cap \bfT(m)$ if $n$~has a parent~$m$
    in~$\calT$, and $A = \emptyset$ if $n$~is the root of~$\calT$;
  \item we set $B = \bfT(n) \setminus A$;
    \item we set $C = \{\tuple{v,w} \in E \mid
    \tuple{v,w} \in \bfT(n)^2 \setminus A^2\}$;
  \item if~$X$~is a set of vertices, then $\chi(X) = \{\chi(v) \mid v
    \in X\}$, and if~$X$~is a set of edges, then $\chi(X) =
    \{\tuple{\chi(v),\chi(w)} \mid \tuple{v,w} \in X\}$.
\end{itemize}
\end{dfn}

These constructions are illustrated in Figure~\ref{fig:etendue-0}, which
displays a graph~$G$, a~tree decomposition $\calD =
\tuple{\calT,\bfT}$ of~$G$, a proper $\calD$-coloring of~$G$, and its
associated succinct tree decomposition.

\begin{figure}[t]
\centering
  \begin{tikzpicture}[scale=0.4,>=stealth]
    \SetGraphUnit{2.5}
    \begin{scope}[VertexStyle/.append style = {minimum size = 15pt, inner sep = 0pt}]
      \Vertex[x=5.1,y=14.75,L={$v_8$}]{1}
      \WE[L={$v_7$}](1){2}
      \WE[L={$v_6$}](2){3}
      \NO[L={$v_5$}](1){4}
      \WE[L={$v_4$}](4){5}
      \WE[L={$v_3$}](5){6}
      \NO[L={$v_2$}](4){7}
      \WE[L={$v_1$}](7){8}
    \end{scope}
    \node[anchor=south] at (2.6,21.3) {Graph $G$};
    
    \Edge[style={<->,thick}](1)(2)
    \Edge[style={<-,thick}](1)(4)
    \Edge[style={->,thick}](2)(4)
    \Edge[style={<->,thick}](2)(5)
    \Edge[style={<-,thick}](4)(5)
    \Edge[style={<->,thick}](2)(6)
    \Edge[style={->,thick}](3)(6)
    \Edge[style={<-,thick}](5)(6)
    \Edge[style={<->,thick}](4)(7)
    \Edge[style={<-,thick}](5)(8)
    \Edge[style={->,thick}](7)(8)
    \Loop[dist=2cm,dir=WE,style={->,thick}](8)

    \begin{scope}[yshift=-2mm]
    \SetGraphUnit{1.6}

    \placenode{0}{10.6}{a}{1}
    \Vertex[Node,L={$v_5$}]{4}
    \Vertex[Node,L={$v_4$}]{5}
    \Vertex[Node,L={$v_2$}]{7}
    \Edge(4)(5)
    \Edge(4)(7)
    
    \placenode{5.3}{10.6}{b}{2}
    \Vertex[Node,L={$v_4$}]{5}
    \Vertex[Node,L={$v_2$}]{7}
    \Vertex[Node,L={$v_1$}]{8}
    \Edge(5)(8)
    \Edge(7)(8)
    
    \placenode{0}{5.3}{c}{3}
    \Vertex[Node,L={$v_7$}]{2}
    \Vertex[Node,L={$v_5$}]{4}
    \Vertex[Node,L={$v_4$}]{5}
    \Edge(2)(4)
    \Edge(2)(5)
    \Edge(4)(5)
    
    \placenode{5.3}{5.3}{h}{4}
    \Vertex[Node,L={$v_8$}]{1}
    \Vertex[Node,L={$v_7$}]{2}
    \Vertex[Node,L={$v_5$}]{4}
    \Edge(1)(2)
    \Edge(1)(4)
    \Edge(2)(4)
    
    \placenode{0}{0}{d}{5}
    \Vertex[Node,L={$v_7$}]{2}
    \Vertex[Node,L={$v_4$}]{5}
    \Vertex[Node,L={$v_3$}]{6}
    \Edge(2)(5)
    \Edge(2)(6)
    \Edge(5)(6)
    
    \placenode{5.3}{0}{e}{6}
    \Vertex[Node,L={$v_6$}]{3}
    \Vertex[Node,L={$v_3$}]{6}
    \Edge(3)(6)
    
    \placenode{-5.3}{0}{f}{7}
    
    \Edge(E1)(W2)
    \Edge(S1)(N3)
    \Edge(E3)(W4)
    \Edge(S3)(N5)
    \Edge(E5)(W6)
    \Edge(W5)(E7)
    
    \node[anchor=south] at (N1) {\textbf{ROOT}};
    \draw[ultra thick] (NW1) -- (NNW1) -- (NNE1) -- (NE1);
    
    \placelabelnode{17.2}{10.6}{\emptyset}{\{0,1,2\}}{\{(0,1),}{(1,0),(2,0)\}}{a}{1}
    \placelabelnode{23.6}{10.6}{\{1,2\}}{\{0\}}{\{(0,0),}{(0,2),(1,0)\}}{b}{2}
    \placelabelnode{17.2}{5.3}{\{0,2\}}{\{1\}}{\{(1,0),}{(1,2),(2,1)\}}{c}{3}
    \placelabelnode{23.6}{5.3}{\{0,1\}}{\{2\}}{\{(0,2),}{(1,2),(2,1)\}}{g}{4}
    \placelabelnode{17.2}{0}{\{1,2\}}{\{0\}}{\{(0,1),}{(0,2),(1,0)\}}{d}{5}
    \placelabelnode{23.6}{0}{\{0\}}{\{2\}}{\{(2,0)\}}{}{e}{6}
    \placelabelnode{10.8}{0}{\emptyset}{\emptyset}{\emptyset}{}{f}{7}
    
    \Edge(E1)(W2)
    \Edge(S1)(N3)
    \Edge(E3)(W4)
    \Edge(S3)(N5)
    \Edge(E5)(W6)
    \Edge(W5)(E7)
    
    \node[anchor=south] at (N1) {\textbf{ROOT}};
    \draw[ultra thick] (NW1) -- (NNW1) -- (NNE1) -- (NE1);
    
    \node[anchor=north] at (2.6,-4.2) {Tree decomposition $\calD$ of $G$};
    \node[anchor=north] at (20.35,-4.2) {$(\chi,\calD)$-succinct tree decomposition of $G$};
    \end{scope}
    
    \node[anchor=south] at (20.35,21.3) {$\calD$-proper coloring of $G$:};
    \node[anchor=south] at (17.35,19.5) {$v_1 \mapsto  0$};
    \node[anchor=south] at (17.35,18.3) {$v_2 \mapsto  1$};
    \node[anchor=south] at (17.35,17.1) {$v_3 \mapsto  0$};
    \node[anchor=south] at (17.35,15.9) {$v_4 \mapsto  2$};
    \node[anchor=south] at (23.35,19.5) {$v_5 \mapsto  0$};
    \node[anchor=south] at (23.35,18.3) {$v_6 \mapsto  2$};
    \node[anchor=south] at (23.35,17.1) {$v_7 \mapsto  1$};
    \node[anchor=south] at (23.35,15.9) {$v_8 \mapsto  2$};
  \end{tikzpicture}
\caption{Graph, tree decomposition, proper coloring and succinct tree decomposition}
\label{fig:etendue-0}\label{ex-xTD-0}
\end{figure}
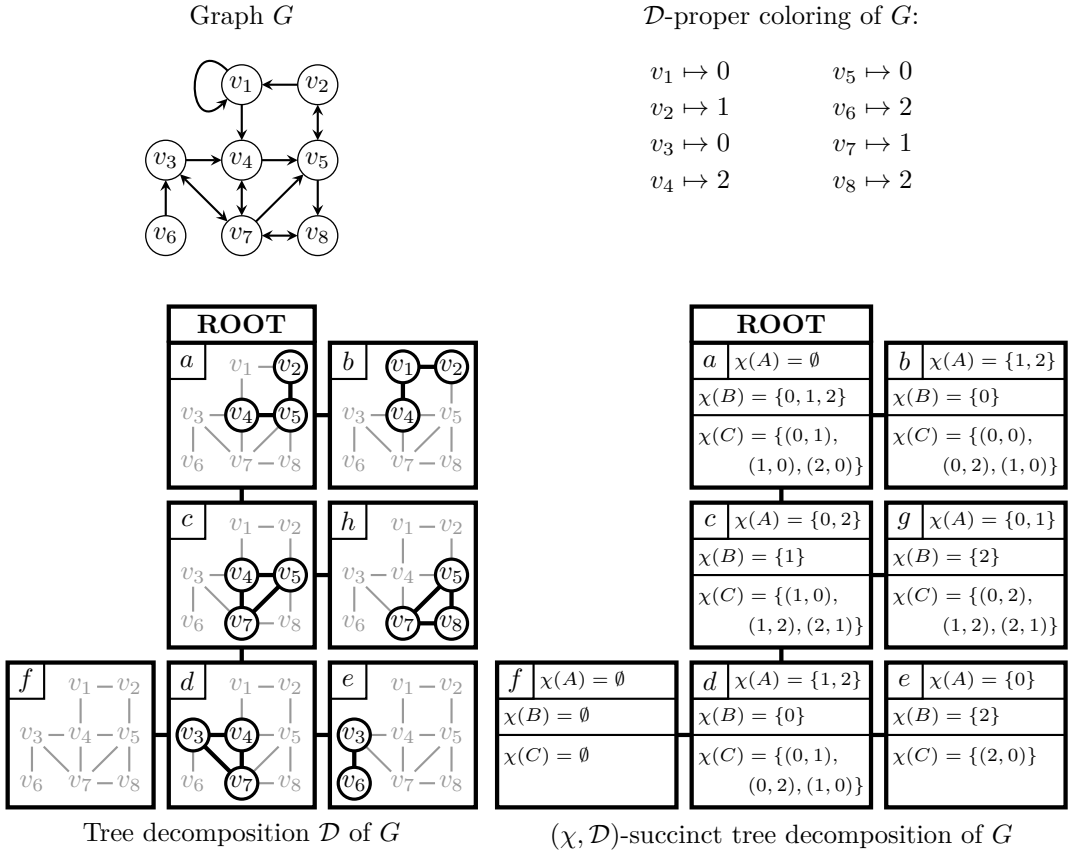

Observe that, given a tree decomposition~$\calD$, there always exist
$\calD$-colorings~$\chi$ of~$G$ and an associated $(\chi,\calD)$-succinct tree
decomposition. They are typically computed from~$\calD$ in a top-down
fashion.
Furthermore, note that a graph~$G$ may have several tree
decompositions~$\calD$ of width~$\kappa$ and, for each of them,
several proper $\calD$-colorings. Hence, $G$~may have several succint
tree decompositions. Yet, from all of them we are able to
reconstruct~$G$ (up~to graph isomorphism), and therefore to check
whether $G$ satisfies the formula~$\varphi$.
A~more precise and powerful version of this statement is the theorem
stated below, which is a variant of the versions of Courcelle's theorem
of~\cite[Section 11.4]{FG06} and~\cite[Section 3.3]{Kre08}.

\begin{thm}\label{thm:courcelle}
  Fix a positive integer $\kappa$ and an \MSO formula $\varphi$.
There exists a (bottom-up, deterministic) tree
automaton~$\mathcal{A}_{\varphi,\kappa}$ such that, for all graphs $G
= \tuple{V,E}$ and all succinct tree decompositions $\calT$ of~$G$ of
width~$\kappa$, $G$~satisfies~$\varphi$ if, and only~if,
$\mathcal{A}_{\varphi,\kappa}$~accepts~$\calT$.
\end{thm}

Making Theorem~\ref{thm:courcelle} useful further requires
being able to compute succinct tree decompositions efficiently.
This is possible thanks to the following result, which is
proven in~\cite{EJT10}:
\begin{lem}
  \label{lemma:puredecomposition}
  Let $G$ be a graph of size~$N$ and tree-width~$\kappa$. We~can
  construct in space $\calO(c(\kappa) \log_2 N)$ an ordered binary tree
  decomposition of~$G$ of size at most $2N$,
  width at most $4 \kappa+3$, and height at most $c(\kappa) \cdot
  (\log_2(N)+1)$, where $c(\kappa)$ only depends
  on~$\kappa$.
\end{lem}

\section{Towards a dynamic algorithm}
\label{sec-stratview}
\looseness=-1
In this section, we focus on making Courcelle's theorem dynamic.
We~fix an input (bottom-up, deterministic) tree automaton~$\calA$ and
an input graph~$G$, and we assume that we have a succinct tree
decomposition~$\calT$ of~$G$.  We~transform the language-theoretic problem of
checking whether $\calA$ accepts~$\calT$ into a Dyck reachability
problem.  More precisely, we~build a graph~$\Gamma_G$ and establish a
correspondence between some Dyck paths in this graph and the (accepting)
runs of~$\calA$ on~$\calT$.

\subsection{State progression}
\label{subsec:reductiontodyck}

A first step towards our goal consists in
performing a sequence of transformations on~$\calT$.

\begin{dfn}
  Let $\calT$ be an ordered tree with $N$ nodes.
  The \emph{post order} on $\calT$ is defined as the
  linear order $\prec$ such that:
  \begin{itemize}
   \item if $m$ is a strict ancestor of $n$, then $n \prec m$;
   \item if an internal node $n$ has a left child $m_1$ and a right child $m_2$,
   then $m_1$ and its descendants are all smaller than $m_2$ and its descendants
   (for the order $\prec$).
  \end{itemize}
  There exists a unique labeling $\posti\colon T \to \{1,\ldots,N\}$,
  which we call \emph{post index},
  such that $n \preccurlyeq m \Leftrightarrow \posti(n) \leqslant \posti(m)$.
  We also commonly denote by $n_i$ the unique node of $\calT$ such that
  $i = \posti(n_i)$.
  
  We further call \emph{bottom-up progression} of $\calT$
  the sequence $\calS_0,\ldots,\calS_N$ of subsets of vertices of $\calT$ defined by
  $\calS_i = \{n \mid \posti(n) \leqslant i \text{ and } \posti(m) > i \text{ for all strict ancestors } m \text{ of } n\}$.
\end{dfn}
Observe that, by construction, we always have $\calS_0 = \emptyset$.

Figure~\ref{fig:etendue-1} presents a binary ordered tree~$\calT$, 
post indices labeling, 
and bottom-up progression.  Observe that $\calT$ is the
(succinct) tree decomposition presented in
Figure~\ref{fig:etendue-0}.

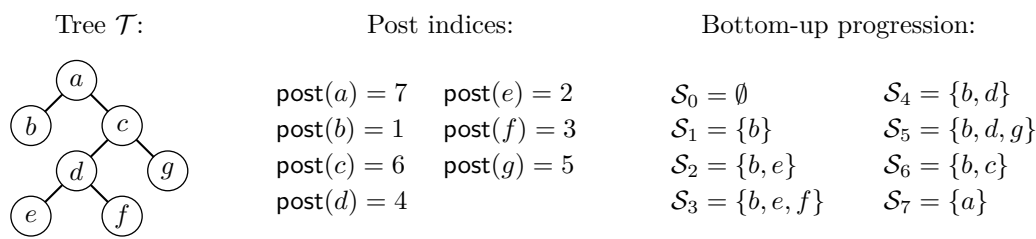
\begin{figure}[t]
\centering
\begin{tikzpicture}[scale=0.4,>=stealth]
\SetGraphUnit{1.5}
\begin{scope}[VertexStyle/.append style = {minimum size = 15pt, inner sep = 0pt}]
\Vertex[x=0.25,y=0,L={$a$}]{a}
\SOWE[L={$b$}](a){b}
\SOEA[L={$c$}](a){c}
\SOEA[L={$g$}](c){d}
\SOWE[L={$d$}](c){e}
\SOWE[L={$e$}](e){f}
\SOEA[L={$f$}](e){g}
\end{scope}

\node[anchor=north] at (1,2.5) {Tree $\calT$:};

\Edge[style={thick}](a)(b)
\Edge[style={thick}](a)(c)
\Edge[style={thick}](c)(d)
\Edge[style={thick}](c)(e)
\Edge[style={thick}](e)(f)
\Edge[style={thick}](e)(g)
\node[anchor=north west] at (9.5,2.5) {Post indices:};
\node[anchor=south west] at (6.5,-1.2) {$\posti(a) = 7$};
\node[anchor=south west] at (6.5,-2.4) {$\posti(b) = 1$};
\node[anchor=south west] at (6.5,-3.6) {$\posti(c) = 6$};
\node[anchor=south west] at (6.5,-4.8) {$\posti(d) = 4$};
\node[anchor=south west] at (12,-1.2) {$\posti(e) = 2$};
\node[anchor=south west] at (12,-2.4) {$\posti(f) = 3$};
\node[anchor=south west] at (12,-3.6) {$\posti(g) = 5$};

\node[anchor=north west] at (20.5,2.5) {Bottom-up progression:};
\node[anchor=south west] at (19.5,-1.2) {$\calS_0 = \emptyset$};
\node[anchor=south west] at (19.5,-2.4) {$\calS_1 = \{b\}$};
\node[anchor=south west] at (19.5,-3.6) {$\calS_2 = \{b,e\}$};
\node[anchor=south west] at (19.5,-4.8) {$\calS_3 = \{b,e,f\}$};

\node[anchor=south west] at (26.5,-1.2) {$\calS_4 = \{b,d\}$};
\node[anchor=south west] at (26.5,-2.4) {$\calS_5 = \{b,d,g\}$};
\node[anchor=south west] at (26.5,-3.6) {$\calS_6 = \{b,c\}$};
\node[anchor=south west] at (26.5,-4.8) {$\calS_7 = \{a\}$};
\end{tikzpicture}
\caption{Ordered tree, depth-first traversal and bottom-up progression}
\label{fig:etendue-1}
\end{figure}

When $\calT$ is binary and has height $h$, then the bottom-up
progression enjoys some conciseness and smoothness
properties, which we state below.

\begin{lem}
\label{lem:progression}
Let $\calT$ be a binary tree with $N$ nodes, of height $h$, and
let $(\calS_i)_{0 \leqslant i \leqslant N}$ be the bottom-up progression of
$\calT$. For all $i \geqslant 1$, it holds that:
\begin{itemize}
 \item the set $\calS_i$ is of cardinality $2 h$ or less;
 \item if the node $n_i$ is a leaf, then $\calS_i = \calS_{i-1} \cup \{n_i\}$;
 \item if the node $n_i$ is an internal node, with children $m_1$ and $m_2$, then
 both $m_1$ and $m_2$ belong to $\calS_{i-1}$, and $\calS_i = \calS_{i-1} \setminus \{m_1,m_2\} \cup \{n_i\}$.
\end{itemize}
\end{lem}

\begin{proof}
First, since $i < \posti(m)$ for all strict ancestors $m$ of $n_i$, it
comes at once that $\calS_i \setminus \calS_{i-1} = \{n_i\}$.
Furthermore, a node $n$ belongs to $\calS_{i-1} \setminus \calS_i$ if,
and only~if, $\posti(n) < i$, $n_i$ is an ancestor of~$n$, and
$\posti(m) \geqslant i$ for all strict ancestors~$m$ of~$n$.  Since we have
$\posti(x) < i < \posti(y)$ for all strict descendants~$x$ and all
strict ancestors~$y$ of~$n_i$, it~follows that
$\calS_{i-1} \setminus \calS_i$ consists of the children of~$n_i$ only
(if~they exist).

Finally, assume that there exist two nodes $x \prec x'$ in $\calS_i$ that
are at the same height, and let~$y$ be the parent of~$x$.
By definition, we know that $\posti(x) < \posti(x') \leqslant i < \posti(y)$,
hence $x'$ belongs to the right subtree of~$y$, i.e., $x'$ is the right sibling of~$x$.
It~follows that $\calS_i$ contains at most two nodes at each level,
whence $|\calS_i| \leqslant 2 h$.
\end{proof}

The above notion of bottom-up progression also leads to the notion
of \emph{state progression}.
\begin{dfn}
Let $\calT$ be a labeled binary ordered tree with $N$ nodes,
let $(\calS_i)_{0 \leqslant i \leqslant N}$
be its bottom-up progression, and let $\calA$ be a (deterministic, bottom-up)
automaton.
Let~$\rho$ be the (unique) run of~$\calA$ over~$\calT$.
We~call \emph{state progression} of~$\calA$ on~$\calT$ the
sequence $(\rho\!\restriction_{\calS_i})_{0 \leqslant i \leqslant N}$,
where $\rho\!\restriction_{\calS_i}$ denotes the restriction of~$\rho$ to
the domain~$\calS_i$.
\end{dfn}

Figure~\ref{fig:etendue-1} presented a tree~$\calT$ and a bottom-up
progression of~$\calT$.  Figure~\ref{fig:etendue-2} presents a tree
automaton~$\calA$, a~labeling of~$\calT$, the~(rejecting) run~$\rho$
of~$\calA$ on~$\calT$ and the associated state progression.  We~omit
to represent the restriction $\rho\!\restriction_{\calS_0}$ since
$\calS_0$ is empty.

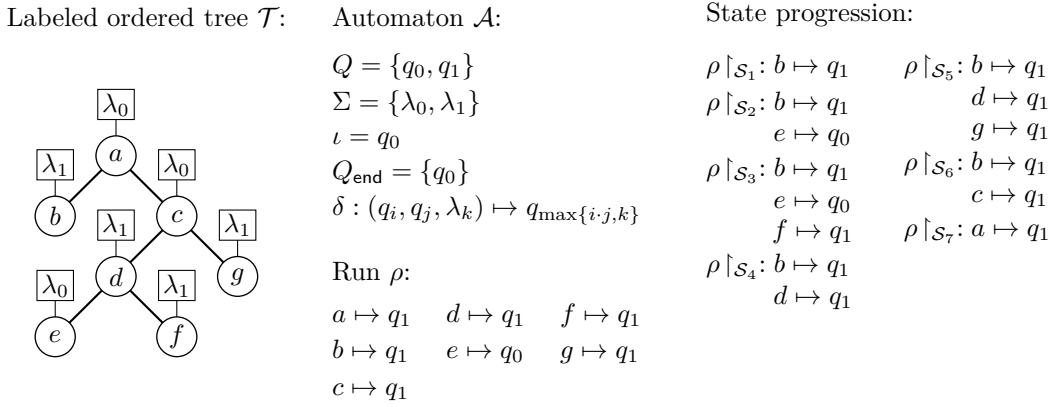
\begin{figure}[t]
\begin{center}
\begin{tikzpicture}[scale=0.4,>=stealth]
\SetGraphUnit{2}
\begin{scope}[xshift=2mm]
\begin{scope}[VertexStyle/.append style = {minimum size = 15pt, inner sep = 0pt}]

\putlabel{0}{-1}{0}
\putlabel{-2}{-3}{1}
\putlabel{2}{-3}{0}
\putlabel{0}{-5}{1}
\putlabel{4}{-5}{1}
\putlabel{-2}{-7}{0}
\putlabel{2}{-7}{1}

\Vertex[x=0,y=-1,L={$a$}]{a}
\SOWE[L={$b$}](a){b}
\SOEA[L={$c$}](a){c}
\SOEA[L={$g$}](c){d}
\SOWE[L={$d$}](c){e}
\SOWE[L={$e$}](e){f}
\SOEA[L={$f$}](e){g}
\end{scope}

\node[anchor=south] at (1,3) {Labeled ordered tree $\calT$:};
\end{scope}

\Edge[style={thick}](a)(b)
\Edge[style={thick}](a)(c)
\Edge[style={thick}](c)(d)
\Edge[style={thick}](c)(e)
\Edge[style={thick}](e)(f)
\Edge[style={thick}](e)(g)

\node[anchor=south west] at (7,3) {Automaton $\calA$:};
\node[anchor=south west] at (7,1.2) {$Q = \{q_0,q_1\}$};
\node[anchor=south west] at (7,0) {$\Sigma = \{\lambda_0,\lambda_1\}$};
\node[anchor=south west] at (7,-1.1) {$\iota = q_0$};
\node[anchor=south west] at (7,-2.3) {$Q_{\END} = \{q_0\}$};
\node[anchor=south west] at (7,-3.6) {$\delta : (q_i,q_j,\lambda_k) \mapsto q_{\max\{i \cdot j,k\}}$};

\node[anchor=south west] at (7,-5.6) {Run $\rho$:};
\node[anchor=south west] at (7,-7) {$a \mapsto q_1$};
\node[anchor=south west] at (7,-8.2) {$b \mapsto q_1$};
\node[anchor=south west] at (7,-9.4) {$c \mapsto q_1$};
\node[anchor=south west] at (10.75,-7) {$d \mapsto q_1$};
\node[anchor=south west] at (10.75,-8.2) {$e \mapsto q_0$};
\node[anchor=south west] at (14.5,-7) {$f \mapsto q_1$};
\node[anchor=south west] at (14.5,-8.2) {$g \mapsto q_1$};

\begin{scope}[xshift=-2mm]
\node[anchor=south west] at (19.5,3) {State progression:};
\node[anchor=south west] at (19.5,1.2) {$\rho\!\restriction_{\calS_1} : b \mapsto q_1$};
\node[anchor=south west] at (19.5,0) {$\rho\!\restriction_{\calS_2} : b \mapsto q_1$};
\node[anchor=south west] at (21.72,-1) {$e \mapsto q_0$};
\node[anchor=south west] at (19.5,-2.2) {$\rho\!\restriction_{\calS_3} : b \mapsto q_1$};
\node[anchor=south west] at (21.72,-3.2) {$e \mapsto q_0$};
\node[anchor=south west] at (21.67,-4.2) {$f \mapsto q_1$};
\node[anchor=south west] at (19.5,-5.4) {$\rho\!\restriction_{\calS_4} : b \mapsto q_1$};
\node[anchor=south west] at (21.72,-6.4) {$d \mapsto q_1$};

\node[anchor=south west] at (26,1.2) {$\rho\!\restriction_{\calS_5} : b \mapsto q_1$};
\node[anchor=south west] at (28.22,0.2) {$d \mapsto q_1$};
\node[anchor=south west] at (28.22,-0.8) {$g \mapsto q_1$};
\node[anchor=south west] at (26,-2) {$\rho\!\restriction_{\calS_6} : b \mapsto q_1$};
\node[anchor=south west] at (28.22,-3) {$c \mapsto q_1$};
\node[anchor=south west] at (26,-4.2) {$\rho\!\restriction_{\calS_7} : a \mapsto q_1$};
\end{scope}
\end{tikzpicture}
\end{center}
\caption{Labeled tree, tree automaton, run and associated state progression}
\label{fig:etendue-2}
\end{figure}

\label{local-rules} For all $i \geqslant 1$, recall that
Lemma~\ref{lem:progression} states that there exists a unique node
${n_i\in\calS_i \setminus \calS_{i-1}}$, and that either $n_i$~is a
leaf or both of its children belong to~$\calS_{i-1}$.  The~functions
$\rho\!\restriction_{\calS_i}$ can therefore be computed in a
step-wise manner once the automaton $\calA =
\tuple{Q,\Sigma,\iota,Q_{\END},\delta}$ is fixed.  More precisely,
and denoting by $\lambda$ the labeling function of $\calT$, we have:
\begin{itemize}
 \item if $i = 0$, then $\calS_0 = \emptyset$, hence
   $\rho\!\restriction_{\calS_0}$ is the empty-domain function;
 \item if $1 \leqslant i$ and $n_i$ is a leaf, then
   $\rho\!\restriction_{\calS_i} =
   \Pi_i(\rho\!\restriction_{\calS_{i-1}},\lambda(n_i))$, where
   \[
   \Pi_i(\varphi,\gamma) \colon n\in\calS_i \mapsto \begin{cases} \varphi(n) &
     \text{if $n \in \calS_{i-1}$;} \\ \delta(\iota,\iota,\gamma) &
     \text{if $n = n_i$;}\end{cases}
   \]
 \item if $1 \leqslant i$ and $n_i$ is an internal node with children $m_1$
   and $m_2$, then $\rho\!\restriction_{\calS_i} =
   \Pi_i(\rho\!\restriction_{\calS_{i-1}},\lambda(n_i))$, where
   \[
   \Pi_i(\varphi,\gamma) \colon n\in \calS_i \mapsto \begin{cases}
     \varphi(n) & \text{if $n \in \calS_{i-1}$;}
     \\ \delta(\varphi(m_1),\varphi(m_2),\gamma) & \text{if $n =
       n_i$.}\end{cases}
   \]
\end{itemize}
We will rely on this step-wise computation in the following section.

\subsection{Reduction to the Dyck reachability problem}
\label{subsec:reductiontodyck-2}

In this section, we~present the reduction of
$\Dyn\GraphMSO_{\kappa,\varphi}$ to a Dyck reachability problem on an
acyclic labeled graph. Our reduction is such that any update (of the
edges) in the input graph corresponds to a simple update of the acyclic graph.
As we explain in Section~\ref{section:overall}, this reduction proves
Theorem~\ref{mainresult}.

\subsubsection{The Dyck reachability problem in acyclic graphs}

Before presenting our reduction, we first define Dyck reachability
problems, then recall briefly some results about their dynamic
complexity in the case of acyclic graphs: in such graphs, context-free
graph queries, and therefore Dyck reachability problems, belong to the
dynamic complexity class \Dyn\FO~\cite{MVZ15,WS07}.

\begin{dfn}
  Let $G = \tuple{V,E,L}$ be a labeled graph, with set of labels $L$,
  and with edge set $E\subseteq V^2\times L$. Let $v_1$ and~$v_2$ be
  two marked vertices of~$G$.  We assume that $L$~can be partitioned
  as $L = L^+ \uplus L^- \uplus \{\bullet\}$, where
  $L^+$ and $L^-$ are in bijection with each other,
  and $\bullet$~is a fresh ``neutral''
  label symbol, and that a bijection $\lambda^+ \mapsto
  \lambda^-$ from $L^+$ to $L^-$ is given.

  The labeling on edges induces in a direct way a labeling on paths
  in~$G$. The \emph{Dyck reachability problem} asks whether there
  exists a path~$\pi$ (in~the graph~$G$) from $v_1$ to~$v_2$, such that
  $\pi$ is labeled with a string in the language~$\bfD$ of \emph{Dyck
    words} built over the grammar: $S \to \varepsilon \mid S \cdot
  \bullet \cdot S \mid S \cdot \lambda^+ \cdot S \cdot \lambda^- \cdot S$,
  where $\lambda^+$ ranges over the set~$L^+$.
\end{dfn}

While~\cite{WS07} assumed a constant-size label set (which is not
  the case here), the result of~\cite{WS07} can be generalized, as
  stated below (see~\cite{arxiv} for a proof).

\begin{lem}
\label{lem:Dyck}
The Dyck reachability problem in acyclic graphs is solvable
in \Dyn\FO (under the assumption that
updates consist in adding or deleting individual labeled edges),
using as only auxiliary predicate a $4$-ary predicate
$\Delta(x_1,y_1,x_2,y_2)$ defined by:
\begin{quote}
``There exists a path $\varpi_1$ from
$x_1$ to $y_1$ with a label $\lambda_1$ and there exists a path $\varpi_2$
from $x_2$ to $y_2$ with a label $\lambda_2$ such that
$\lambda_1 \cdot \lambda_2$ is a Dyck word''.
\end{quote}
\end{lem}

\subsubsection{Reduction}

We fix an \MSO formula $\varphi$, and let~$\calA_{\varphi,4\kappa+3} =
\tuple{Q,\Sigma,\iota,Q_{\END},\delta}$, which we simply name~$\calA$ in the rest of this section, be the (deterministic,
bottom-up) tree automaton of Theorem~\ref{thm:courcelle}.

We now describe our
transformation of any subgraph~$G$ of~$G_\star = \tuple{V,E_\star}$
into an acyclic labeled graph~$\Gamma_G$ for the Dyck reachability
problem.
Let $\calD_\star = \tuple{\calT_\star,\mathbf{T}_{\star}}$, where
$\calT_\star = \tuple{\calN_\star,\calE_\star}$ be a
tree decomposition of~$G_\star$ of width $4\kappa+3$, as defined in
Lemma~\ref{lemma:puredecomposition},
and let~$\calN_\star$ be the set of nodes of~$\calT_\star$.  
In~addition, let~$\chi$ be a $\calD_\star$-coloring function of~$V$,
and let $(\calS_i)_{0 \leqslant i \leqslant N}$ be the bottom-up
progression of~$\calT_\star$.

Let~$G = \tuple{V,E}$ be a subgraph of~$G_\star$.  Observe that
$\calD_\star$ is also a tree decomposition of~$G$, since $E
  \subseteq E_\star$.  Hence, there exists a labeling function
$\Lambda_G \colon \calN_\star \to \Sigma$ that
identifies~$\calT_\star$ with a $(\chi,\calD_\star)$-succinct binary tree
decomposition~$\calT_G$ of~$G$.  In~particular, $(\calS_i)_{0 \leqslant i
  \leqslant N}$ is also the bottom-up progression of~$\calT_G$.
Let~$\rho_G$ be the run of~$\calA$ over~$\calT_G$.  We~want
to construct a graph~$\Gamma_G$ in order to identify the state
progression $(\rho_G\!\restriction_{\calS_i})_{0 \leqslant i \leqslant N}$
with a (Dyck) path in~$\Gamma_G$.

\subparagraph*{A~naive construction.}
As a first try, we~let the
vertices of this graph be all
  pairs $(i,\pi)$, where $0 \le i \le N$, and $\pi\colon \calS_i \to
  Q$ is intended to represent the state progression at
  step~$i$. Following the local computation described
  page~\pageref{local-rules}, we include an edge $(i-1,\pi) \to
  (i,\pi')$ when $\pi' = \Pi_i(\pi,\Lambda_G(n_i))$, where $n_i$ is
  the unique node in $\calS_i \setminus \calS_{i-1}$. Then, obviously,
  if $\pi_{\INIT}$ is the unique function from~$\emptyset$ to~$Q$, the
  path $\rho_G$ is accepting if and only if, in this naive graph, the
  unique path from $(0,\pi_{\INIT})$ to $(N,\pi)$ is such that $\pi$
  maps the~unique element of~$\calS_N$ (namely, the~root of the tree)
  onto~$Q_\END$.

  While this naive construction is correct, it is not
  suitable in a dynamic complexity perspective: adding or removing an edge
  in~$G$ may affect many edges in the
  above graph. However, we~show below that it can only affect edges 
  at a single level (index~$i$);
  using Dyck constraints, we~then adapt the construction above to
  have updates of~$G$ only impact one edge of our graph.

  The idea is illustrated on Figure~\ref{fig:cluster}.  Assume that
  both upper (naive) graphs represent a parameterized function~$f$
  with parameters~$\gamma$~(left) and~$\gamma'$~(right): there is an
  edge~$\alpha_x \to \beta_y$ in the left graph whenever $\beta_y =
  f(\alpha_x,\gamma)$, and an edge $\alpha_x \to \beta_y$ in the right
  graph whenever $\beta_y = f(\alpha_x,\gamma')$.  Replacing the value
  of~$\gamma$ with~$\gamma'$, i.e., transforming the left graph into the
  right~one, requires many edge deletions and insertions.
  
  \looseness=-1
  We circumvent this problem by using Dyck paths: the~value of
  $f(\cdot,\gamma)$ is computed thanks to the Dyck path labeled with
  $\alpha_x^+ \cdot \bullet \cdot \alpha_x^-$, which goes from
  $\alpha_x$ to $\beta_y = f(\alpha_x,\gamma)$.  Hence, changing the
  value of~$\gamma$ to~$\gamma'$ amounts to replacing the edge $\circ
  \xrightarrow{\bullet} \gamma$ with the new edge $\circ
  \xrightarrow{\bullet} \gamma'$.

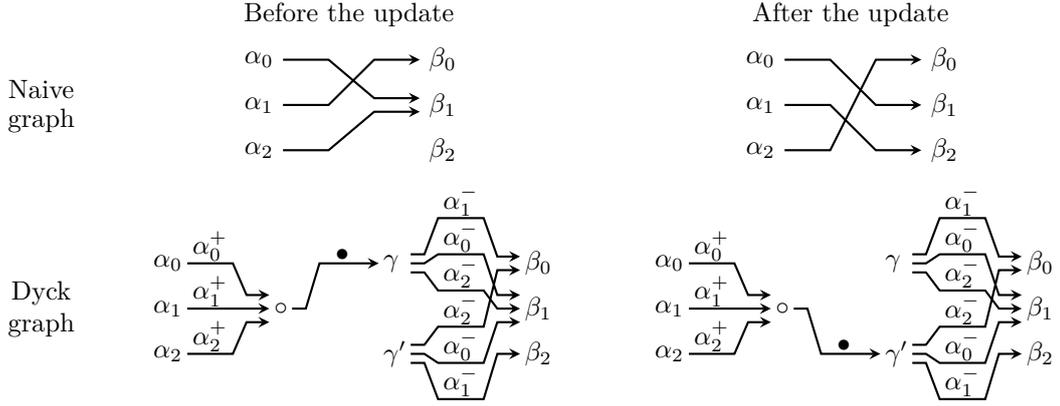
\begin{figure}[t]
  \begin{center}
    \begin{tikzpicture}[scale=0.6,>=stealth]      
      \node at (-3,5.85) {Naive};
      \node at (-3,5.15) {graph};
      
      \node[anchor=east] at (2.3,6.5) {$\alpha_0$};
      \node[anchor=east] at (2.3,5.5) {$\alpha_1$};
      \node[anchor=east] at (2.3,4.5) {$\alpha_2$};
      \draw[thick,->] (2.3,6.5) -- (3.3,6.5) -- (4.3,5.65) -- (5.3,5.65);
      \draw[thick,->] (2.3,5.5) -- (3.3,5.5) -- (4.3,6.5) -- (5.3,6.5);
      \draw[thick,->] (2.3,4.5) -- (3.3,4.5) -- (4.3,5.35) -- (5.3,5.35);
      \node[anchor=west] at (5.3,6.5) {$\beta_0$};
      \node[anchor=west] at (5.3,5.5) {$\beta_1$};
      \node[anchor=west] at (5.3,4.5) {$\beta_2$};
      
      \node[anchor=east] at (13.3,6.5) {$\alpha_0$};
      \node[anchor=east] at (13.3,5.5) {$\alpha_1$};
      \node[anchor=east] at (13.3,4.5) {$\alpha_2$};
      \draw[thick,->] (13.3,6.5) -- (14.3,6.5) -- (15.3,5.5) -- (16.3,5.5);
      \draw[thick,->] (13.3,5.5) -- (14.3,5.5) -- (15.3,4.5) -- (16.3,4.5);
      \draw[thick,->] (13.3,4.5) -- (14.3,4.5) -- (15.3,6.5) -- (16.3,6.5);
      \node[anchor=west] at (16.3,6.5) {$\beta_0$};
      \node[anchor=west] at (16.3,5.5) {$\beta_1$};
      \node[anchor=west] at (16.3,4.5) {$\beta_2$};
      
      \node at (-3,1.35) {Dyck};
      \node at (-3,0.65) {graph};

      \node[anchor=east] at (0.3,2) {$\alpha_0$};
      \node[anchor=east] at (0.3,1) {$\alpha_1$};
      \node[anchor=east] at (0.3,0) {$\alpha_2$};
      
      \draw[thick,->] (0.2,2) -- (1.2,2) -- (1.5,1.3) -- (2,1.3);
      \draw[thick,->] (0.2,1) -- (2,1);
      \draw[thick,->] (0.2,0) -- (1.2,0) -- (1.5,0.7) -- (2,0.7);
      \node[anchor=south] at (0.7,1.85) {$\alpha_0^+$};
      \node[anchor=south] at (0.7,0.85) {$\alpha_1^+$};
      \node[anchor=south] at (0.7,-0.15) {$\alpha_2^+$};
      
      \node[anchor=west] at (1.9,1) {$\circ$};
      \draw[thick,->] (2.5,1) -- (2.8,1) -- (3.1,2) -- (4.4,2);
      \node[anchor=south] at (3.6,1.85) {$\bullet$};
      
      \node[anchor=west] at (4.3,2) {$\gamma$};
      \node[anchor=west] at (4.3,0) {$\gamma'$};
      
      \draw[thick,->] (5.1,2.2) -- (5.4,2.2) -- (5.7,3) -- (6.7,3) -- (7,2.15) -- (7.5,2.15);
      \draw[thick,->] (5.1,2) -- (5.4,2) -- (5.7,2.2) -- (6.7,2.2) -- (7,1.3) -- (7.5,1.3);
      \draw[thick,->] (5.1,1.8) -- (5.4,1.8) -- (5.7,1.4) -- (6.7,1.4) -- (7,1) -- (7.5,1);
      \draw[thick,->] (5.1,0.2) -- (5.4,0.2) -- (5.7,0.6) -- (6.7,0.6) -- (7,1.85) -- (7.5,1.85);
      \draw[thick,->] (5.1,0) -- (5.4,0) -- (5.7,-0.2) -- (6.7,-0.2) -- (7,0.7) -- (7.5,0.7);
      \draw[thick,->] (5.1,-0.2) -- (5.4,-0.2) -- (5.7,-1) -- (6.7,-1) -- (7,0) -- (7.5,0);
      \node[anchor=south] at (6.2,2.85) {$\alpha_1^-$};
      \node[anchor=south] at (6.2,2.05) {$\alpha_0^-$};
      \node[anchor=south] at (6.2,1.25) {$\alpha_2^-$};
      \node[anchor=south] at (6.2,0.45) {$\alpha_2^-$};
      \node[anchor=south] at (6.2,-0.35) {$\alpha_0^-$};
      \node[anchor=south] at (6.2,-1.15) {$\alpha_1^-$};
      
      \node[anchor=west] at (7.4,2) {$\beta_0$};
      \node[anchor=west] at (7.4,1) {$\beta_1$};
      \node[anchor=west] at (7.4,0) {$\beta_2$};
      
      \node[anchor=east] at (11.3,2) {$\alpha_0$};
      \node[anchor=east] at (11.3,1) {$\alpha_1$};
      \node[anchor=east] at (11.3,0) {$\alpha_2$};
      
      \draw[thick,->] (11.2,2) -- (12.2,2) -- (12.5,1.3) -- (13,1.3);
      \draw[thick,->] (11.2,1) -- (13,1);
      \draw[thick,->] (11.2,0) -- (12.2,0) -- (12.5,0.7) -- (13,0.7);
      \node[anchor=south] at (11.7,1.85) {$\alpha_0^+$};
      \node[anchor=south] at (11.7,0.85) {$\alpha_1^+$};
      \node[anchor=south] at (11.7,-0.15) {$\alpha_2^+$};
      
      \node[anchor=west] at (12.9,1) {$\circ$};
      \draw[thick,->] (13.5,1) -- (13.8,1) -- (14.1,0) -- (15.4,0);
      \node[anchor=south] at (14.6,-0.15) {$\bullet$};
      
      \node[anchor=west] at (15.3,2) {$\gamma$};
      \node[anchor=west] at (15.3,0) {$\gamma'$};
      
      \draw[thick,->] (16.1,2.2) -- (16.4,2.2) -- (16.7,3) -- (17.7,3) -- (18,2.15) -- (18.5,2.15);
      \draw[thick,->] (16.1,2) -- (16.4,2) -- (16.7,2.2) -- (17.7,2.2) -- (18,1.3) -- (18.5,1.3);
      \draw[thick,->] (16.1,1.8) -- (16.4,1.8) -- (16.7,1.4) -- (17.7,1.4) -- (18,1) -- (18.5,1);
      \draw[thick,->] (16.1,0.2) -- (16.4,0.2) -- (16.7,0.6) -- (17.7,0.6) -- (18,1.85) -- (18.5,1.85);
      \draw[thick,->] (16.1,0) -- (16.4,0) -- (16.7,-0.2) -- (17.7,-0.2) -- (18,0.7) -- (18.5,0.7);
      \draw[thick,->] (16.1,-0.2) -- (16.4,-0.2) -- (16.7,-1) -- (17.7,-1) -- (18,0) -- (18.5,0);
      \node[anchor=south] at (17.2,2.85) {$\alpha_1^-$};
      \node[anchor=south] at (17.2,2.05) {$\alpha_0^-$};
      \node[anchor=south] at (17.2,1.25) {$\alpha_2^-$};
      \node[anchor=south] at (17.2,0.45) {$\alpha_2^-$};
      \node[anchor=south] at (17.2,-0.35) {$\alpha_0^-$};
      \node[anchor=south] at (17.2,-1.15) {$\alpha_1^-$};
      
      \node[anchor=west] at (18.4,2) {$\beta_0$};
      \node[anchor=west] at (18.4,1) {$\beta_1$};
      \node[anchor=west] at (18.4,0) {$\beta_2$};

      \path[use as bounding box] (3.75,7.5);
      \node[] at (3.75,7.5) {Before the update};
      \node[] at (14.75,7.5) {After the update};

      \end{tikzpicture}
  \end{center}
  \caption{Using Dyck paths saves many changes when the input graph is updated \label{fig:cluster}}
\end{figure}

\subparagraph*{The refined construction.}

A ``nominal'' vertex of the graph~$\Gamma_G$ is a pair $(i,\pi)$,
where $0 \leqslant i \leqslant N$ and $\pi$~is a function $\pi \colon \calS_i
\to Q$, intended to represent the state progression at step~$i$.
Labels of~$\Gamma_G$ are the pairs $(i,\pi)^+$ and $(i,\pi)^-$, and
the neutral label $\bullet$.  We~write~$\pi_{\INIT}$ for the unique
function from~$\emptyset$ to~$Q$.  Now, for $1 \le i \le N$, let $n_i$
be the unique node in~$\calS_i \setminus \calS_{i-1}$, and recall that
$\rho_G\!\restriction_{\calS_i}
= \Pi_i(\rho_G\!\restriction_{\calS_{i-1}},\Lambda_G(n_i))$.  We
therefore add the following edges and vertices:
\begin{itemize}
 \item vertices $(i-1)$ and $(i-1,\gamma)$ for all $\gamma \in \Sigma$;
 \item edges $(i-1,\pi) \xrightarrow{(i-1,\pi)^+} (i-1)$ and
   $(i-1,\gamma) \xrightarrow{(i-1,\pi)^-} (i,\pi')$ for all
   $\gamma \in \Sigma$ and all $\pi \colon \calS_{i-1} \to Q$, where $\pi'
   \colon \calS_i \to Q$ is such that $\pi' = \Pi_i(\pi,\gamma)$;
 \item one neutral edge $(i-1) \xrightarrow{\bullet} (i-1,\gamma)$, where $\gamma = \Lambda_G(n_i)$.
\end{itemize}

Finally, observe that $n_N$ is the root of~$\calT_\star$, that
$\calS_N = \{n_N\}$, and that the run~$\rho_G$ is accepting if,
and only~if, $\rho_G(n_N) \in Q_\END$.  Hence, we~complete the
construction of~$\Gamma_G$ by adding a last state~$\top$ and neutral
edges $(N,\pi) \xrightarrow{\bullet} \top$ for those functions $\pi
\colon \calS_N \to Q$ such that $\pi(n_N) \in Q_\END$.

This construction is both sound and complete, and well-behaved under
modifications of~$G$, as outlined by the following results.

\begin{proposition}
\label{pro:sound-complete}
The automaton $\calA$ accepts the labeled tree~$\calT_G$ if, and
only~if, there exists a Dyck path in~$\Gamma_G$ from the vertex
$(0,\pi_{\INIT})$ to the vertex~$\top$.
\end{proposition}

\begin{proof}
  A path from $(0,\pi_{\INIT})$ to  a vertex
   $(N,\varpi)$, where $\varpi$ is a function $\calS_N \to Q$,
  is Dyck if, and only~if, it~uses only sub-paths of the form
  $(i-1,\pi) \xrightarrow{(i-1,\pi)^+} (i-1) \xrightarrow{\bullet}
  (i-1,\gamma) \xrightarrow{(i-1,\pi)^-} (i,\pi')$ with $\gamma =
  \Lambda_G(n_i)$ and $\pi' = \Pi_i(\pi,\gamma)$.  Hence, such a Dyck
  path exists if, and only~if, $\varpi = \rho_G(n_N)$, where $\rho_G$
  is the run of $\calA$ on~$\calT_G$, in which case the intermediate
  vertices of the form~$(i,\pi)$ are the
  vertices~$(i,\rho_G\!\restriction_{\calS_i})$.  The result follows
  immediately.
\end{proof}

\begin{proposition}
\label{pro:small-change}
Let $e = (v,w)$ be an edge of the maximal graph~$G_\star$, and let $G$
and $G'$ be two subgraphs of~$G_\star$ such that $G'$ is obtained by
adding (resp. deleting) the edge~$e$ to~$G$.
The~graph~$\Gamma_{G'}$ is obtained by deleting an edge~$e_1$
from~$\Gamma_G$ and inserting another edge~$e_2$ instead. Furthermore,
both edges~$e_1$ and~$e_2$ are \FO-definable in terms of~$e$,
of~$\Gamma_G$, and of some auxiliary precomputed predicates.
\end{proposition}

\begin{proof}
We only deal here with insertion of an edge.
We associate with the edge $e = \tuple{v,w}$ of~$G_\star$ 
a mapping $\add_e \colon \Sigma \to \Sigma$
defined by
$\add_e \colon \tuple{\chi(A),\chi(B),\chi(C)} \to
\tuple{\chi(A),\chi(B),\chi(C) \cup \{\tuple{\chi(v),\chi(w)}\}}$.
Then, let $n$ be the top-most node of~$\calT_\star$ such that both~$v$
and~$w$ belong to~$\mathbf{T}(n)$, and let ${i = \posti(n)}$.
Lemma~\ref{lem:progression} states that $n \in \calS_i \setminus
\calS_{i-1}$.  Hence, the labeling functions $\Lambda_G$
and~$\Lambda_{G'}$ coincide on all nodes $m \neq n$, and we have
$\Lambda_{G'}(n) = \add_e(\Lambda_G(n))$.  Consequently,
the~graph~$\Gamma_{G'}$ is obtained from~$\Gamma_G$ in two consecutive
steps:
\begin{enumerate}
\item we delete the only outgoing edge, of the vertex $(i-1)$, which
  is a neutral edge of the form $(i-1) \xrightarrow{\bullet}
  (i-1,\gamma)$;
\item we add the new edge $(i-1) \xrightarrow{\bullet}
  (i-1,\add_e(\gamma))$.
\end{enumerate}

The case of deletion is analogous, and requires using a mapping
$\del_e$ similar to $\add_e$.
Since the mappings $e \mapsto i$, $(e,\gamma) \mapsto
\add_e(\gamma)$ and $(e,\gamma) \mapsto \del_e(\gamma)$ can be
precomputed, it~follows that both the edge~$e_1$ that we deleted
from~$\Gamma_G$ and the edge~$e_2$ that we inserted instead can be
computed with \FO formulas.
\end{proof}

\section{Overall complexity analysis}
\label{section:overall}
In this section, we analyze the complexity of our dynamic algorithm.
While adequate notions of reduction do exist in dynamic complexity
  (see e.g.~\cite{PI97,HI02}),
  our reduction does not satisfy all criteria, so we need to compute
  the complexity of our algorithm by hand.

First, denoting by $\calV$ and $\calL$ the vertex set and the label
set of $\Gamma_G$, Lemma~\ref{lem:Dyck} states that the Dyck reachability
problem in $\Gamma_G$ can be solved by using \FO update formulas
\emph{over the universe $\calV \uplus \calL$}.
However, we need \FO formulas over the universe~$V$ of our
\MSO model checking problem, i.e., $V$~is the vertex set of the input graph.
Hence, we must embed $\calV \uplus \calL$ into a set of tuples
of elements of~$V$ of finite arity.

Lemma~\ref{lemma:puredecomposition} states that $\mathcal{T}_\star$ is of
height at most $c(\kappa) \cdot (\log_2(N)+1)$, where $N = |V|$, and
Lemma~\ref{lem:progression} proves that
$|\calS_i| \leqslant 2 c(\kappa) \cdot (\log_2(N)+1)$ for all $i \leqslant N$.
It follows that $|\calV| = 1 + \sum_{i=0}^N |Q|^{|\calS_i|} =
\calO(N^{2 c(\kappa) \log_2(|Q|) + 1})$, which is
polynomial in~$|V|$. Likewise, $|\calL|$~is polynomial in~$|V|$, and
therefore $\calV \uplus \calL$ can be embedded into some set~$V^k$,
where $k$ is a large enough integer (which depends only on~$\kappa$
and on the \MSO formula~$\varphi$).

In the end, during the precomputation phase, 
the algorithm successively computes:
\begin{enumerate}
\item a binary rooted tree decomposition $\mathcal{D}_\star =
  \tuple{\mathcal{T}_\star,\mathbf{T}_{\star}}$ of the maximal
  graph $G_\star = \tuple{V,E_\star}$, of width $4\kappa+3$,
  such as described in Lemma~\ref{lemma:puredecomposition};
\item a (bottom-up, deterministic) tree automaton
$\calA_{\varphi,4\kappa+3}$ such as defined in Courcelle's theorem;
\item a $\calD_\star$-coloring function $\chi$, a $(\chi,\calD_\star)$-succinct tree
  decomposition of $G_\star$, and a bottom-up progression
  $(\calS_i)_{0 \leqslant i \leqslant N}$ of $\calT_\star$;
\item the vertices, labels and edges of the graph~$\Gamma_{G_E}$,
where $G_E$ is the initial input graph;
\item an embedding $\calV \uplus \calL \mapsto V^k$;
\item mappings $e \mapsto i$, $(e,\gamma) \mapsto \add_e(\gamma)$
and $(e,\gamma) \mapsto \del_e(\gamma)$ mentioned in the proof of
Proposition~\ref{pro:small-change};
\item the value of the auxiliary predicate $\Delta$
(mentioned in Lemma~\ref{lem:Dyck}) on
$\Gamma_{G_E}$.
\end{enumerate}

\newpage

\begin{lem}
\label{lem:precomputation-in-logspace}
Each of these 7 steps can be performed in \LOGSPACE.
\end{lem}

\begin{proof}
The formula $\varphi$ and the integer~$\kappa$ are fixed.
Hence, Lemma~\ref{lemma:puredecomposition} proves that the step~1 can be
performed in \LOGSPACE, and the step~2 is completed in constant time.
Since performing the steps~3--6 in \LOGSPACE is straightforward,
it remains to deal with the step~7.

Let $\varpi$ be a path with label $\lambda$ in $\Gamma_G$.
We say that $\varpi$ is a \emph{Dyck prefix} path
if $\lambda$ is a prefix of a Dyck word (which may be $\lambda$ itself)
and if its proper prefixes are not Dyck words;
that $\varpi$ is a \emph{Dyck suffix} path
if $\lambda$ is a suffix of a Dyck word
and if its proper suffixes are not Dyck words;
that $\varpi$ is a \emph{minimal Dyck} path if
$\varpi$ is both a Dyck prefix and a Dyck suffix path.

Minimal non-empty Dyck paths are paths of the form
$(i-1) \xrightarrow{\bullet} (i-1,\gamma)$, $(\ell,\pi) \xrightarrow{\bullet} \top$,
and $(i-1,\pi) \xrightarrow{(i-1,\pi)^+} (i-1) \xrightarrow{\bullet}
(i-1,\gamma) \xrightarrow{(i-1,\pi)^-} (i,\pi')$.
Since Dyck prefix paths are the prefixes of minimal Dyck paths,
and Dyck suffix paths are the suffixes of minimal Dyck paths,
they have length at most 2.
Furthermore, every Dyck path is a product of non-empty minimal Dyck paths.
In addition, if $\varpi_1$ and $\varpi_2$ are paths with labels
$\lambda_1$ and $\lambda_2$ such that $\lambda_1 \cdot \lambda_2$
is a Dyck word, then there exists factorisations
$\varpi_1 = \varpi_1^{\INIT} \cdot \varpi_1^{\END}$ and
$\varpi_2 = \varpi_2^{\INIT} \cdot \varpi_2^{\END}$ such that
$\varpi_1^{\INIT}$ and $\varpi_2^{\END}$ are Dyck paths,
$\varpi_1^{\END}$ is a Dyck prefix and $\varpi_2^{\INIT}$ is a Dyck suffix.

Hence, the predicate $\Delta(x_1,y_1,x_2,y_2)$ holds if, and only~if,
there exists vertices $z_1$ and~$z_2$ and paths~$\varpi_1$ (from~$z_1$ to~$y_1$)
and $\varpi_2$ (from~$x_2$ to~$z_2$) with labels $\lambda_1$ and~$\lambda_2$ such that:
\begin{itemize}
 \item $\varpi_1$ and $\varpi_2$ are of length at most 2,
 and $\lambda_1 \cdot \lambda_2$ is a Dyck word;
 \item there exists Dyck paths from $x_1$ to $z_1$ and from $z_2$ to $y_2$.
\end{itemize}
Finally, note that every vertex of $\Gamma_G$ is the source of at most
one minimal Dyck path. Consequently, for any two vertices~$x$ and~$y$
of~$\Gamma_G$, checking if there exists a Dyck path from~$x$ to~$y$
can be done in \LOGSPACE, and $\Delta$ can be computed in \LOGSPACE
too.
\end{proof}

We sum up the above results as follows.
First, we perform a \LOGSPACE precomputation,
and construct a graph $\Gamma_G$ whose vertex, label and edge sets
can be represented as predicates of finite arity on the universe $V$.
Then, during the update phases, whenever introducing or deleting an
edge $e$ in $G$, we replace one edge of $\Gamma_G$ by another one,
and these edges are identified by precomputed \FO formulas taking the
edge $e$ into account, as stated in Proposition~\ref{pro:small-change}.
Consequently, and since the Dyck reachability problem is in \Dyn\FO,
updating the edge-membership predicate
of $\Gamma_G$ and the auxiliary predicate $\Delta$,
which is useful for solving the Dyck reachability problem in $\Gamma_G$,
can be done with \FO formulas.
Finally, deciding whether $G$ satisfies the formula $\varphi$,
i.e., whether there exists a suitable Dyck path in $\Gamma_G$,
can be done using directly the auxiliary predicate $\Delta$,
which completes the proof of Theorem~\ref{mainresult}.

\section{Conclusion}
\label{sec-concl}
We developed a dynamic algorithm for checking a (fixed) \MSO formula
over (evolving) subgraphs of a given graph of bounded
tree-width. A~natural extension of this work 
would consist in getting rid of the hypothesis that there exists a
maximal graph~$G_\star$ of which the graphs under scrutiny are
subgraphs.
There are two main obstacles for this to be achieved in our approach:
first, we would need to be able to
dynamically compute tree decompositions of ``moderate'' width of our
dynamic graphs; then, we~would have to adapt the structure of our
graph~$\Gamma_G$ to take into account these evolving tree decompositions.

Another direction of research, which was successfully put into practice
in~\cite{arxiv} when dealing with the particular case of parity games,
would consist, given an input formula $\varphi = \exists X.\ \varphi'(X)$
(starting with an existential quantifier), to~compute a witness~$X$
of the satisfiability of~$\varphi'$.

\end{document}